\documentclass[a4paper,10pt]{article}

\usepackage{amsmath,amsthm,bm}
\usepackage[utopia]{mathdesign}
\usepackage{graphicx,subcaption}
\usepackage[table]{xcolor}
\usepackage{slashbox}
\usepackage{authblk}
\usepackage[UKenglish]{isodate}

\usepackage[top=3cm,left=2.5cm]{geometry}
\usepackage[italian,english]{babel} 
\usepackage[normalem]{ulem} % per sottolineature, cancellature

\newcounter{pp}[section]

\newtheorem{theorem}{Theorem}[section]
\newtheorem{lemma}{Lemma}[section]
\newtheorem{proposition}{Proposition}[section]

\theoremstyle{remark}

\newcommand{\bA}{\mathbf{A}}
\newcommand{\ba}{\mathbf{a}}
\newcommand{\bB}{\mathbf{B}}
\newcommand{\bb}{\mathbf{b}}
\newcommand{\bc}{\mathbf{c}}

\newcommand{\be}{\mathbf{e}}

\newcommand{\bI}{\mathbf{I}}

\newcommand{\bl}{\bm{\ell}} %{\mbox{\boldmath $\ell$}}
\newcommand{\bL}{\mathbf{L}}
\newcommand{\mm}{\mathbf{m}}
\newcommand{\bM}{\mathbf{M}}
\newcommand{\bn}{\mathbf{n}}
\newcommand{\bN}{\mathbf{N}}
\newcommand{\bQ}{\mathbf{Q}}
\newcommand{\vel}{\mathbf{v}}
\newcommand{\bu}{\mathbf{u}}

\newcommand{\R}{\mathbb{R}}
\newcommand{\bR}{\mathbf{R}}
\newcommand{\Rey}{\mathcal{R}}
\newcommand{\bS}{\mathbb{S}}
\newcommand{\bST}{\mathbf{S}}
\newcommand{\bSs}{\bS^{\text{sym}}}

\newcommand{\bT}{\mathbf{T}}
\newcommand{\bTT}{\mathbb{T}}

\newcommand{\bv}{\mathbf{v}}
\newcommand{\bx}{\mathbf{x}}

\newcommand{\V}{V_2}

\newcommand{\Cl}{\mathcal{C}}

\newcommand{\qp}{\boxtimes}
\newcommand{\WD}{\mathcal{D}}

\renewcommand{\rho}{\varrho}

\DeclareMathOperator{\tp}{\otimes}
\DeclareMathOperator{\kp}{\boxtimes}
\DeclareMathOperator{\tr}{tr}

\newcommand{\comm}[1]{{\color{black}#1}}
\def\msfO{{\mathsf O}}

\def\msfS{{\mathsf S}}
\newcommand{\SO}[1]{
\mbox{$\msfS \msfO(#1)$}%\mbox{$\mcalS \mcalO(#1)$}
}
\newcommand{\Og}[1]{
\mbox{$\msfO(#1)$}%\mbox{$\mcalO(#1)$}
}
\def\essotre{{\SO{3}}}
\def\essodue{{\SO{2}}}
\def\Otre{{\Og{3}}}
\def\Odue{{\Og{2}}}

\def\real{{\mathbb{R}}}

\def\ex{\be_{1}}
\def\ey{\be_{2}}
\def\ez{\be_{3}}
\def\lx{\bl_{1}}
\def\ly{\bl_{2}}
\def\lz{\bl_{3}}

\def\diadlxx{{\lx\otimes\lx}}
\def\diadlyy{{\ly\otimes\ly}}
\def\diadlzz{{\lz\otimes\lz}}

\def\transp{^{\mathrm T}}
\def\suchcol{\,\colon \, }
\def\dd{\mathrm d}
\title{\textbf{Determination of the symmetry classes\\of orientational ordering tensors}}

\author[1]{Stefano S.\,Turzi \thanks{\texttt{stefano.turzi@polimi.it}}}
\author[2]{Fulvio Bisi}%

\affil[1]{\small{Dipartimento\,di\,Matematica, Politecnico\,di\,Milano, 
Piazza\,Leonardo\,da\,Vinci,\,32~~20133\,Milano\,(Italy)}}
\affil[2]{\small{Dipartimento di Matematica, Universit\`a di Pavia, Via Ferrata, 1~~27100 Pavia (Italy)}}

\cleanlookdateon
\date{\today}

\begin{document}
\maketitle

\begin{abstract}
The orientational order of nematic liquid crystals is traditionally studied by means of the second-rank ordering tensor 
$\bS$. When this is calculated through experiments or simulations, the symmetry group of the phase is not known 
\emph{a-priori}, 
but needs to be deduced from the numerical realisation of $\bS$, which is affected by numerical errors. 
There is no generally accepted procedure to perform this analysis. Here, we provide a new algorithm suited 
to identifying the symmetry group
of the phase. As a by product, we prove that there are only five phase-symmetry classes of the second-rank
ordering tensor and give 
a canonical representation of $\bS$ for each class. The nearest tensor of the assigned symmetry is determined 
by group-projection. In order to test our procedure, we generate uniaxial and biaxial phases in a system
of interacting particles, endowed with $D_{\infty h}$ or $D_{2h}$ symmetry, which mimic the outcome of 
Monte-Carlo simulations. The actual symmetry of the phases is correctly identified, along with the optimal
choice of laboratory frame. 
\end{abstract}

\section{Introduction}

The orientational order of an ensemble of molecules is a key feature in complex fluids made of anisotropic molecules,
e.g. liquid crystals. For example, the phase of a liquid crystal affects some rheological and optical properties 
of the material, such as viscosity coefficients and refracting index. From a mathematical viewpoint, the phase is
a macroscopic manifestation of the point-group symmetry of the mesoscopic orientational order of the molecules.
The precise quantification of the notion of order requires the introduction of the orientational probability 
density function of the molecules. It is impractical to study this function in its full generality from 
a mathematical perspective. Furthermore, only its very first moments are amenable of experimental investigation.
For these reasons the orientational probability density is usually truncated at the second-rank level and this 
defines the second-rank ordering tensor $\bS$. Usually, the matrix entries of $\bS$ are 
considered to capture correctly the most important features of the mesoscopic order. 

The final output of a molecular dynamics or a Monte Carlo simulation of a liquid crystal compound is given in terms 
of the orientation of the molecular frames of reference, for all molecules. The ordering tensor  $\bS$ is then 
obtained by averaging over all molecular orientations. However, the computations have to be carried out with respect 
to an arbitrarily chosen laboratory frame. By contrast, key physical information as phase symmetry, directors 
and order parameters are readily accessible only when the laboratory axes are properly chosen in agreement 
with the yet unknown underlying symmetry of the orientational distribution. Therefore, the experimental 
or simulation data need to be analysed and refined in order to capture the physical features of the system 
at the meso-scale and there is no standard method to perform this analysis.

The main motivation of the present work is to provide such a systematic procedure to determine the symmetry class
(the ``phase'' of the system), the symmetry axes (the ``directors''), and the scalar order parameters 
of a liquid crystal compound whose second-rank ordering tensor is obtained through experiments or simulations. 

Experimental or numerical errors are a further source of complications and may hinder the correct identification
of the phase symmetry\comm{, even in the simplest cases}. For instance, \comm{a uniaxial order could be
described naively as a phase in which rod-like molecules are substantially aligned parallel to a fixed direction, 
identified by a}
well defined director; \comm{in such phase} most of the entries of $\bS$ \comm{ought to} vanish, 
\comm{but the presence of errors can make all the entries generally non-zero. Furthermore,
an unwise choice of the laboratory axes may be the cause of several non-vanishing entries.}
%the interdependence of which may be hidden.} 
When dealing with such an ordering tensor, it may not be immediately evident whether 
\comm{these non-vanishing entries reveal an intrinsic lack of uniaxial phase symmetry or are a consequence of one --or possibly both-- of the issues described.}

Our strategy takes inspiration from a similar problem in Elasticity where the main concern is the identification 
of the linear elastic tensor of a particular material symmetry. This problem has been intensively studied 
by a number of authors in the last decades. We refer to \cite{ForteVianello,2011Slawinski} for a historical overview.
In particular, we adopt similar mathematical techniques and in this respect we have found the following papers 
particularly illuminating \cite{ForteVianello,2011Slawinski,2004Slawinski,1987Wadhawan,1963Toupin,1998Geymonat}.

This paper is organised as follows. Sec.~\ref{sec:background} reviews the theoretical background 
on orientational order parameters; namely, the spherical and Cartesian definitions of orientational ordering
tensors are discussed. The specific case of \emph{second-rank} ordering tensors is developed in 
Sec.~\ref{sec:second-rankOPs}. 
\comm{The definitions we provide here best fit the group-theoretic analysis put forward in the rest of the paper,
and allow taking a non-standard view on this topic, by describing the ordering tensor in terms of a linear map
in the space of symmetric, traceless second rank tensors;} an analogous approach 
is found in Refs.~\cite{2011turzi,2015chillb,2015chill,07Rosso}. 
In Secs.~\ref{sec:second-rankOPs_symmetry} and~\ref{sec:second-rankOPs_symmetryclasses} we define the notion
of symmetry-class of an ordering tensor and prove that it is only possible to distinguish five phase-symmetry classes
at the second-rank level. A finer identification of the phase group of a non-polar liquid crystal requires higher 
rank ordering tensors. The following Sec.~\ref{sec:identification} deals with the identification of the closest 
ordering tensor that belongs to a given symmetry class. To this end, we introduce the invariant projection 
on a chosen symmetry group, define the distance of the raw ordering tensor from one of the five symmetry classes 
and provide a canonical representation of the ordering tensor in each symmetry class. After all these mathematical ingredients are established, 
 we describe the algorithm for the determination of the effective phase in Sec.~\ref{sec:identification_effectivephase}. 
 Sec.~\ref{sec:examples} contains the discussion of two paradigmatic examples where the algorithm is put into practice and
 Sec.~\ref{sec:conclusions} summarises the results.

\subsection{Notations.} \label{sec:notations}

\comm{For the reader's convenience, let us give a brief description of some notational conventions used throughout the paper.

\begin{enumerate}
\item Vectors in the linear space $W$  isomorphic to $\real^3$ are denoted by boldface small letters ($\ba, \bb, \bc, \dots,\bu, \vel,\dots$). 
After choosing an orthonormal basis $\{\ex,\ey,\ez\}$, the coordinate of a vector $\vel$ are denoted with the same plain letter, and a subscript (in general $i=1,\dots,3$) distinguishes the coordinates: $\vel=v_1 \ex+ v_2 \ey + v_3 \ez$.  This avoids confusion between $\vel_1$, which the vector 1 in a list of vectors, and $v_1$ which is the first coordinate of vector $\vel$: $v_1 = \vel \cdot \ex$.
Unit vectors in the laboratory reference frame play a special role, therefore we will denote them by $\bl_{\xi}$,
with $\xi = x,y,z$.
 
 \item Second-rank tensors, i.e. linear % operators of $W$, 
 maps $\bL \colon W \to W$ in the linear space $W$, are denoted
 by boldface capital letters ($\bA, \bB, \dots, \allowbreak \bR, \bST, \bT, \dots$); $\bI$ is the identity tensor.
 
 \item \label{sec:notationsdiad}
 The tensor (or dyadic) product in $W$ between vectors $\ba$ and $\bb$ is a second-rank tensor such that
 $(\ba \tp \bb)\bv = (\bb\cdot\bv)\ba$ for every vector $\bv$, 
 where the dot ${}\dot{}$ denotes the standard 
 inner (scalar) product (i.e. the matrix representative of the dyadic product is
 $(\ba \tp \bb)_{ij} = a_i b_j$ $(i,j = 1,\ldots, 3)$ in an orthonormal basis).
 
 \item \label{sec:innerT}
 The scalar product between two second-rank tensors $\bT$, $\bL$ is defined as
 \begin{equation} \label{eq:dot_tensors}
 \bT\cdot\bL = \tr(\bT\transp \bL) . 
 \end{equation}
 
 \item
 Second-rank tensors can be endowed with the structure of a linear space, denoted by %$\mcalL$
 $L(W)$; linear maps $\bTT \colon L(W) \to L(W)$ in such space are denoted
 by ``blackboard'' capital letters ($\bS, \bTT, \dots$).
 
 \item \label{sec:tensordiad}
 The square tensor product $\qp$ between second-rank tensors is to be interpreted as a 
 tensor dyadic product: % $(\bA \qp \bB) \bl = (\bB \cdot \bL) \bA$.
 \[
 (\bL \qp \bM) \,\bT = (\bM \cdot \bT)\,\bL, \quad \text{ for every tensor } \bT.
 \]

 \item $\Otre$ is the orthogonal group, i.e. the group of all isometries of $W$ ($\real^3$); 
 $\bA \in \Otre \, \Leftrightarrow \, \bA^{-1} =\bA\transp \, \Rightarrow \det \bA = \pm 1$, 
 where the superscript $()\transp$ denotes the transposition.
 
 \item $\essotre$ is the special orthogonal group, i.e. the subgroup of $\Otre$ of elements $\bR$ satisfying 
 $\det(\bR) = 1$. In other words, the group of 3D rotations.
 
 \item Similarly, $\Odue$ is the orthogonal group in two dimensions, and 
 $\essodue$ is the special orthogonal group in two dimensions
 
 \item \label{sec:schoenflies}
 For point symmetry groups, subgroups of $\Otre$, we comply with
 the standard Sch\"{o}nflies notation \cite{PointGroups,2001Michel,McWeeny}.
 Here, we only give a brief description of these groups and refer the reader 
 to the cited references for a more in-depth discussion. The complete list include the seven infinite sequences 
 of the axial groups $C_n$, $C_{nh}$, $C_{nv}$, $S_{2n}$, $D_{n}$, $D_{nd}$, $D_{nh}$ and seven exceptional 
 groups $T$, $T_d$, $T_h$, $O$, $O_h$, $I$, $I_h$. The axial group $C_n$ contains $n$-fold rotational symmetry 
 about an axis, and 
 $D_n$ contains $n$-fold rotational symmetry about an axis and a 2-fold rotation about a perpendicular axis.
 
 The other axial groups are obtained by adding reflections across planes through the main rotation axis, 
 and/or reflection across the plane perpendicular to the axis. In particular, sub-indexes $h$, $v$, and $i$ 
 stands for ``horizontal'', ``vertical'' and ``inversion'' and denotes, respectively, the presence of a 
 mirror reflection perpendicular to rotation axis ($\sigma_h$), a mirror reflection parallel to the 
 rotation axis ($\sigma_v$) and the inversion ($\iota = -\bI$). 
 We recall that $C_1$ is the trivial ``no symmetry'' group; $S_2$ is the group of order 
 two that contains the inversion and is usually written as $C_i$; the group of order 
 two with a single mirror reflection is denoted by $C_{1h}$, $C_{1v}$ or $C_s$. 
 By contrast, the seven exceptional groups contains multiple 3-or-more-fold rotation axes: 
 $T$ is the rotation group of a regular tetrahedron, $O$ is the rotation group of a cube or octahedron and $I$ is the rotation group of the icosahedron. Finally, taking $n \to \infty$ yields the additional continuous groups: 
 $C_{\infty}$, $C_{\infty h}$, $C_{\infty v}$, $D_{\infty}$, $D_{\infty h}$. 
 $C_{\infty}$ is another notation for $\essodue$, 
 and $C_{\infty v}$ is $\Odue$, which can be generated by $C_{\infty}$ and a reflection through 
 any vertical plane containing the vertical rotation axis.
 
 \item \label{sec:ensemble}
The ensemble average of a function $\chi$ with respect to the orientational probability distribution of the molecules is sometimes denoted by angle brackets: $\langle \chi \rangle$.
 
\end{enumerate}
}

\section{General background on orientational order parameters}
\label{sec:background}

\comm{In the first part of the present section, we recall the basics of liquid crystal theory;
the simplest mesogenic molecules have cylindrical symmetry ($D_{\infty h}$, see Sec.~\ref{sec:notations}(\ref{sec:schoenflies}.)),
and if they are basically arranged so as to have their main axis (identified by a unit vector $\mm$ parallel to the 
cylindrical symmetry axis) along one preferred direction, the mesophase is \emph{uniaxial}; however, each molecule 
might slightly deviate from the alignment direction, and in a uniaxial phase this deviation occurs randomly with equal 
probability in any other direction.

Formally speaking, if $f(\mm)$ is the distribution describing the probability that the direction of the 
main axis of the molecule is exactly $\mm$, and considering that $f(\mm) = f(-\mm)$, since the opposite direction
cannot be distinguished due to the symmetry of the molecule, we need to resort to the second moment of the distribution
\[
\bN = \int_{S^2} (\mm \otimes \mm) f(\mm)\ \dd \Omega\,,
\]
where the integration is performed over the unit sphere $S^2$, in other words on all directions (e.g. $\dd \Omega = \sin\beta \dd \beta\ \dd \alpha$, with $\alpha$ and $\beta$ defined as the colatitude and the azimuth in a spherical frame).

Since we are interested in describing \emph{non}-isotropic phases, typically the order tensor 
$\bQ = \bN - \frac{1}{3} \bI$ is used, which is identically zero in the isotropic phase.
$\bQ$ is a symmetric traceless tensor (since $\bN$ is symmetric and $\tr \bN =1$), and a uniaxial phase
corresponds to (at least) two equal eigenvalues for $\bQ$; if we choose a laboratory (orthonormal) reference frame  
$( \lx, \ly, \lz )$ for which $\lz$ is along the preferred direction of the molecule (known as 
the \emph{director} $\bn$), we can write
\begin{equation}\label{eq:Quniaxial}
 \bQ = S \left(\diadlzz - \frac{1}{3}\bI \right) = S \left(\bn \otimes \bn - \frac{1}{3}\bI \right)\,; 
\end{equation}
the scalar $S \in [-\frac{1}{2}, 1]$ is the main uniaxial order parameter, and is actually
the ensemble average $\langle \frac{3}{2} \cos^2\beta - \frac{1}{2} \rangle$;
$S=1$ would describe ``perfect'' alignment along $\bn$\footnote{Actually, 
the phase we have described is a \emph{calamitic}
uniaxial phase. A uniaxial phase can also be \emph{discotic} when $\mm$ is randomly distributed in a
plane orthogonal to $\bn$, in which case ``perfect'' alignment would correspond to $S=-\frac{1}{2}$}.

We point out that the matrix representation of the tensor $\bQ=S(\bn \otimes \bn - \frac{1}{3}\bI)$, which intrinsically
describes a uniaxial phase, can be strongly different from the diagonal form if the reference lab frame is 
poorly chosen (recall that the elements of the matrix representative of the dyadic product $\bn \otimes \bn$
are $n_i  n_j, \,\, i,j = 1,\dots,3$; cf. Sec.~\ref{sec:notations}).

In a generic uniaxial phase, the order parameter does not attain its maximum value, as molecules are not
perfectly aligned along $\bn$.
However, all molecule might deviate from the direction of $\bn$ not entirely in a random way. To 
picture the situation, we can say that in the previous case the molecule are uniformly distributed 
in a circular cone the axis of which has direction $\bn$, although the aperture of the cone 
is small. Under different circumstances, e.g. a
frustration induced by the boundary of the region in which the liquid crystal is confined, the
molecules might be distributed in an elliptical cone, meaning there are two directions orthogonal to 
$\bn$ along which the molecule have maximum and minimum deviations. By properly choosing the laboratory frame,
the order tensor (now having 3 different eigenvalues) can be written as
\[
  \bQ = S \left(\diadlzz - \frac{1}{3}\bI \right) + P (\diadlxx - \diadlyy)\,,
\]
where the additional biaxial order parameter $P = \langle \sin^2\beta\, \cos 2\alpha \rangle$ ranges in $[-1,1]$ 
and vanishes for uniaxial phases. However,
whenever $P \neq 0$ the phase is not purely uniaxial, but has a \emph{biaxial} (orthorhombic) phase symmetry: a symmetry 
lower than that of the molecule.

On a dual point of view, molecules endowed with a $D_{2h}$ symmetry 
are characterised by 3 main axes instead of one; whenever only the main axes $\bm$ are aligned,
the phase has a higher $D_{\infty h}$ (uniaxial)
symmetry; differently, when also the other two axes tend to align respectively in two orthogonal directions,
we obtain a phase with the same $D_{2h}$ (biaxial) symmetry. We omit the detail of the description in the laboratory frame 
(see~\cite{2003Virga,universal,bisi2011}), however pictures in Fig~\ref{fig:D2hmols} show the difference of the two cases.
}

\begin{figure}[ht]
\centering
\begin{subfigure}{0.4\textwidth}
\includegraphics[width=0.95\textwidth]{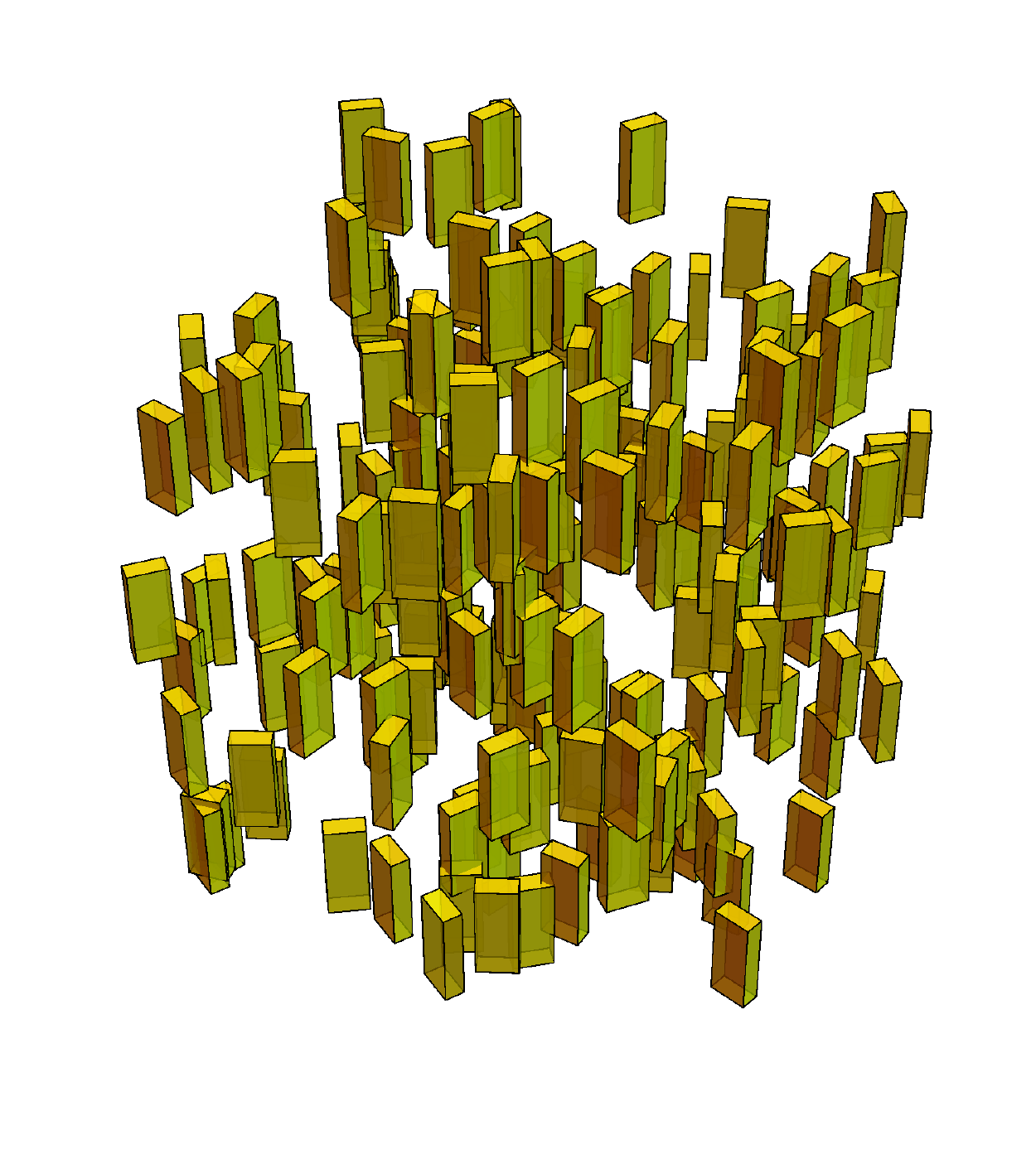} 
\caption{}
\end{subfigure}
\begin{subfigure}{0.4\textwidth}
\includegraphics[width=0.95\textwidth]{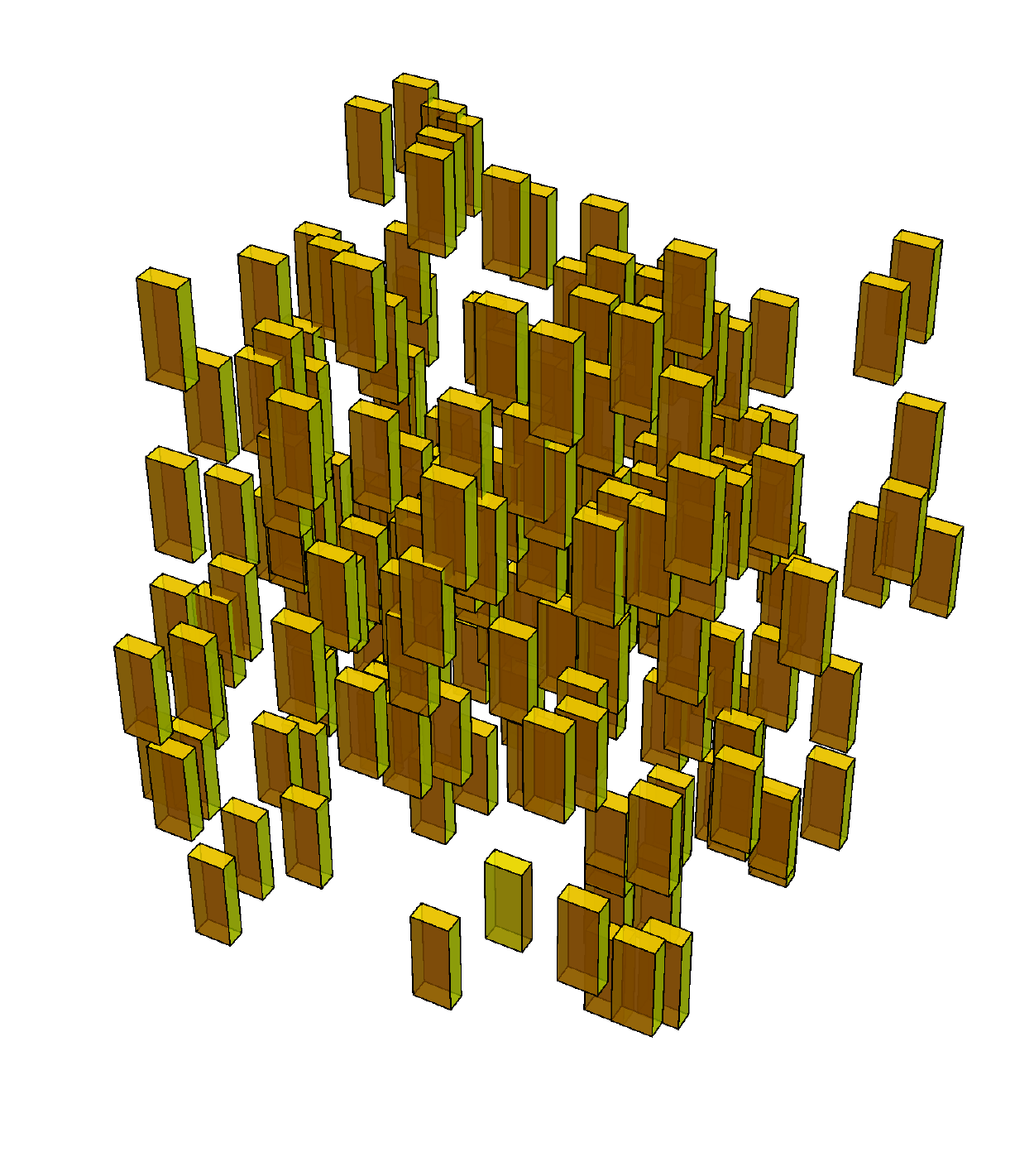} 
\caption{}
\end{subfigure}
\caption{Schematic representations of (a) $D_{\infty h}$ and (b) $D_{2h}$ phase-symmetries made 
with molecules possessing a $D_{2h}$ symmetry.}
\label{fig:D2hmols}
\end{figure}

In general, when no symmetry for the molecule or the phase can be assumed \emph{a-priori}, the orientational distribution of a collection of molecules is most conveniently described in terms 
of a space-dependent probability density function $f(\bx,\bR)$. Here $\bx$ is the space point of the considered molecule 
and $\bR \in \essotre$ is the rotation of a right orthonormal frame set in the molecule, with $(\mm_1, \mm_2, \mm_3)$
unit vectors along the axes, with respect to the laboratory frame of reference, identified by the three mutually 
orthogonal unit vectors $(\bl_1, \bl_2, \bl_3)$. In the following we drop the explicit dependence of $f$ 
on the space point $\bx$ since we are mainly interested in its orientational properties and we will assume 
a continuous dependence on $\bR$. Hence, $f \colon \essotre \to \R_+$ is a continuous function from the group 
of proper rotations
to the non-negative real numbers. Several equivalent description can be given of the rotation matrix $\bR$ that describes 
the orientation of the molecule with respect to the laboratory axes. The matrix $\bR$ is defined as the rotation that 
brings \comm{each} unit vector $\bl_i$ into coincidence with the \comm{corresponding}
molecular unit vector $\mm_i$: \comm{$\bR \bl_i = \mm_i, 
\allowbreak i=1,\dots,3$},
and the matrix entries of $\bR$ are given by the director cosines \comm{$R_{ij} = \bl_i \cdot \bR\bl_j = \bl_i \cdot \mm_j
\allowbreak i,j = 1,\dots,3$}. 
%where the dot stands for the standard scalar product in the three dimensional Euclidean space. 
Equivalent but more intrinsic descriptions of the same matrix are obtained as follows
\begin{equation}
\bR = \sum_{i,j=1}^{3} R_{ij} \bl_i \tp \bl_j, \qquad \text{ or } \qquad
\bR = \sum_{k=1}^{3} \mm_k \tp \bl_k \, ,
\label{eq:rotation_matrix}
\end{equation}
(cf. Sec~\ref{sec:notations})

It is known from group representation theory (more precisely from Peter-Weyl theorem, 
see for instance~\cite{1986Barut,Sternberg}) that the matrix entries of all the irreducible representations 
of the rotation group $\essotre$ form a complete orthonormal set for the continuous functions $f \colon \essotre \to \R$.
Traditionally, this irreducible decomposition is based on the properties of the spherical harmonic functions,
as we now recall. The space $L^2(S^2)$ of square integrable functions over the two-dimensional unit sphere $S^2$
can be decomposed into the infinite direct sum of suitable finite-dimensional vector spaces 
$V_j$: $L^2(S^2) = \bigoplus_j V_j$, 
where $j$ is a non-negative integer. Each $V_j$ is generated by the spherical harmonics $\{Y_{jk}\}$ of rank $j$ 
and has dimension $\dim(V_j) = 2j + 1$. The \comm{irreducible} representation of $\essotre$, given by $\WD^{(j)}$,
is defined 
by assigning the linear maps $\WD^{(j)}(\bR): V_j \to V_j$ such that for all $\psi \in V_j$, $\bR \in \essotre$
\begin{equation}
\WD^{(j)}(\bR)\psi(\bx) = \psi (\bR\transp\bx)
\end{equation}
(where a superscript `$\mathrm{T}$' stands for transpose).
The explicit expressions for this irreducible representations of $\essotre$ are usually known as Wigner rotation matrices \cite{Wigner}.
Hence, $f$ is usually expanded in terms of Wigner rotation matrices \cite{Wigner,Rose,1979zannoni}.

\subsection{Cartesian definition}
\label{sec:background_Cartesian}
However, for many purposes it is more convenient to use an equivalent definition and identify $V_j$
with the space of \emph{traceless symmetric tensors} of rank $j$. \comm{In fact, also traceless symmetric tensors can be used to form a basis for the irreducible representations of $\essotre$
\cite{2011turzi,Wigner}.}
This result is known in other branches of mathematics and physics and is sometimes called 
\emph{harmonic tensor decomposition}~\cite{ForteVianello}. The irreducible representations are then
identified with the non-singular linear maps $D^{(j)}(g):V_j \to V_j$, defined as follows. 
If $(\bv_1, \bv_2, \ldots, \bv_j)$ are a set of $j$ vectors belonging to the three dimension real 
vector space $W\simeq \R^3$,
we define the action of $g \in \essotre$ on $V_j$ as the restriction of the diagonal action of $\essotre$ over 
the tensor product space $W^{\tp j}:=\underbrace{W \tp W \tp \ldots \tp W}_{j\text{ times}}$. Explicitly, we define 
\begin{equation}
D(g)(\bv_1 \tp \bv_2 \tp \ldots \tp \bv_j) = g\bv_1 \tp g\bv_2 \tp \ldots \tp g\bv_j
\label{eq:diagonal_action}
\end{equation}
and then extend by linearity to any tensor $\bT \in W^{\tp j}$. The $D^{(j)}(g)$ is then obtained as the 
restriction of $D(g)$ to $V_j \subset W^{\tp j}$. Peter-Weyl theorem can now be used to show that the 
entries $\sqrt{2j+1}\,D^{(j)}_{pm}(\bR)$, form a purely Cartesian complete orthonormal system for 
the continuous functions in $\essotre$, with respect to its \emph{normalised invariant (or Haar) measure} 
\comm{(in terms of the common Euler angles $\alpha, \beta,\gamma$, such measure is explicitly given by
$\dd \mu = \frac{1}{8\pi^2}\sin\beta \dd \beta \ \dd \alpha \ \dd \gamma$)}. 
The Fourier expansion of the probability distribution function is
\footnote{\comm{Basically, this expansion is the analogue of the Fourier analysis for compact groups.}}

\begin{equation}
f(\bR) = \sum_{j=0}^{+\infty}\,\, (2j+1) \!\!\!\comm{\sum_{p,m=0}^{2j}} f_{pm}^{(j)} D^{(j)}_{pm}(\bR) , 
\end{equation}
where the coefficients are readily obtained via orthogonality from the integrals
\begin{equation}
f_{pm}^{(j)} = \int_{\essotre} D^{(j)}_{pm}(\bR) f(\bR)\ \dd\mu(\bR) \, .
\label{eq:exp_coefficients}
\end{equation}

When studying phase transformations, it is often useful to define an order parameter,
that is a quantity which changes the value on going from one phase to the other and that can therefore
be used to monitor the transition. From a molecular point of view, however, we should describe the
passage from one phase to another in terms of the modifications that this produces in the distribution function.
Therefore, a standard assumption is to identify the order parameters with the expansion coefficients 
\eqref{eq:exp_coefficients} 
(see Refs. \cite{1973Boccara,1979zannoni,1985zannoni,2001luckhurst,LuckhurstBook,bisirossoC2v})

\section{Second-rank order parameters}
\label{sec:second-rankOPs}
In nematic liquid crystals, the expansion of the probability distribution function is usually 
truncated at $j=2$.
The $j=0$ term represents the isotropic distribution, while the $j=1$ terms vanishes for symmetry reasons. 
Therefore, the first non trivial information about the molecular order is provided by the $j=2$ terms 
which then acquire a particular important part in the theory. Higher rank terms are sometimes also studied, 
but this is uniquely done in particular cases (i.e. uniaxial molecules) where simplifying assumptions or 
the symmetry of the problem restrict the number of independent order parameters and the complexity of their calculation. 
Surely, it is already difficult to have an insight about the physical meaning of the $j=2$
order parameters, when no particular symmetry is imposed \cite{2011turziC2h,2012turziD2h,2013TurziSluckin}.

Therefore, in the present paper we will only consider the \emph{second-rank order parameters}, i.e.
with $j=2$, although, at least formally, it is easy to extend our definitions to higher rank ordering tensors.
The invariant space  $\V$ is described as the (five dimensional) space of symmetric and traceless \comm{second-rank} 
tensors\footnote{\comm{In Elasticity theory, 
traceless tensors are sometimes called \emph{deviatoric} tensors. 
Since the common second-rank tensors found in Elasticity are symmetric, $V_2$ is usually referred to as the 'space of 
deviatoric tensors'.}}
on the three dimensional real space $W \simeq \R^3$. \comm{ Let $L(\V)$ be the space of the linear maps $\V \to \V$.}

Given the general definition of Cartesian ordering tensor \eqref{eq:exp_coefficients}, we are led to consider 
traceless symmetric tensor spaces and define second-rank ordering tensor (or order parameter tensor) as the \comm{linear 
map $\bS \in L(\V)$}, such that \cite{2015chillb,2015chill}

\begin{equation}
\bS(\bT) = \int_{\essotre} D^{(2)}(\bR)\bT \,f(\bR)\ \dd\mu(\bR) := \langle D^{(2)}(\bR)\rangle \bT, 
\end{equation}
where the $(j=2)$-irreducible representation matrix $D^{(2)}(\bR)$ acts explicitly by conjugation as follows
\begin{equation}
D^{(2)}(\bR)\bT = \bR \bT \bR\transp .
\label{eq:conjugation}
\end{equation}

It is worth noticing that the 5 $\times$ 5 matrices $D^{(2)}(\bR)$, defined by \eqref{eq:conjugation}, 
yield an irreducible real \emph{orthogonal} representation of $\Otre$ and as such they satisfy 
\begin{equation}
D^{(2)}(\bR)\transp = D^{(2)}(\bR\transp) = D^{(2)}(\bR^{-1})= D^{(2)}(\bR)^{-1}.
\label{eq:D(R)orthogonality} 
\end{equation}
\comm{Since $\Otre$ is the direct product of $\essotre$ and $C_i$, each representation of $\essotre$ splits into two representations of $\Otre$. However, in our application the natural generalisation of \eqref{eq:conjugation} to the $j^{th}$-rank tensors induces us to choose the following representation for the inversion: $D^{(j)}(\iota)= (-1)^j D^{(j)}(\bI)$. Hence, when $j=2$, we obtain $D^{(2)}(\iota) = D^{(2)}(\bI)$.}

It is convenient for the sake of the presentation to introduce an orthonormal basis that describes 
the orientation of the molecule in the 5-dimensional space $\V$. It is natural to build this basis 
on top of the three dimensional orthonormal frame $(\mm_1,\mm_2,\mm_3)$. Therefore, in agreement 
with~\cite{07Rosso,47,48}\footnote{In Ref.~\cite{07Rosso} definitions of $\bM_3$ and $\bM_4$ are swapped}, 
we define 
% Under an abstract mathematical perspective it unnecessary to distinguish between isomorphic spaces. 
% However, it is convenient for the sake of the presentation to distinguish between the linear space $\VM$, 
% generated by the orthonormal frame $(\mm_1,\mm_2,\mm_3)$ which describes the orientation of a generic molecule
% and $\VL$, with basis obtained from the laboratory frame of reference $(\bl_1,\bl_2,\bl_3)$.
% We follow \cite{47,48,07Rosso}, and associate to a generic molecule an orthonormal basis spanning 
% the five-dimensional linear space of traceless symmetric tensors $\VM$ given by 
\begin{subequations}
\begin{align}
\bM_0 & = \sqrt{\frac{3}{2}} \left(\mm_3 \tp \mm_3 - \frac{1}{3}\bI \right) , &
\bM_1 & = \frac{1}{\sqrt{2}} \left(\mm_1 \tp \mm_1 - \mm_2 \tp \mm_2 \right), \\
\bM_2 & = \frac{1}{\sqrt{2}} \left(\mm_1 \tp \mm_2 + \mm_2 \tp \mm_1 \right), &
\bM_3 & = \frac{1}{\sqrt{2}} \left(\mm_2 \tp \mm_3 + \mm_3 \tp \mm_2 \right), \\
\bM_4 & = \frac{1}{\sqrt{2}} \left(\mm_1 \tp \mm_3 + \mm_3 \tp \mm_1 \right).
\end{align}
\label{eq:molorienttensors}
\end{subequations}
The tensors $\{\bM_0,\bM_1,\ldots,\bM_4\}$ are orthonormal with respect to the standard scalar product 
\eqref{eq:dot_tensors}.

Similarly, we define the basis of five symmetric, traceless tensors ${\bL_0,\ldots, \bL_4}$ in terms of 
${\bl_1, \bl_2,\bl_3}$ that are used as laboratory frame of reference. The matrix $D^{(2)}(\bR)$ is 
an irreducible representation of the rotation $\bR$ in the 5-dimensional space of symmetric traceless tensors. 
Specifically, it describes how the ``molecular axis'' $\bM_i$ of a given molecule is rotated with respect 
the laboratory axis:
\begin{equation}
D^{(2)}(\bR) \bL_i = \bM_i, \qquad i = 0,\dots,4,
\end{equation}
to be compared with the similar expression $\bR \bl_i = \mm_i \allowbreak \, (i=1,\dots,3)\,$ 
in the three-dimensional space. Likewise,
components of $D^{(2)}(\bR)$ and Eq.~\eqref{eq:rotation_matrix} become
\begin{equation}
D^{(2)}(\bR)_{ij} = \bL_i \cdot \bM_j , \qquad 
D^{(2)}(\bR) = \sum_{i,j=0}^{4} D^{(2)}(\bR)_{ij} \bL_i \qp \bL_j , \qquad
D^{(2)}(\bR) = \sum_{k=0}^{4} \bM_k \qp \bL_k,
\end{equation}
\comm{where $i,j = 0,\dots,4$}, and 
\comm{$\qp$ stands for the tensor dyadic product (cf. Sec.~\ref{sec:notations}, \eqref{sec:tensordiad}.)} % \\
%\mbox{$(\bL_i \qp \bM_j) \,\bT = (\bM_j \cdot \bT)\,\bL_i $}, for every tensor $\bT$.

Hence, the Cartesian components of the ordering tensor $\bS \in L(\V)$ are \comm{defined as}
\begin{equation}
S_{ij} = \bL_{i} \cdot \bS(\bL_j) = \bL_{i} \cdot \langle \bM_j \rangle,\qquad i,j= 0,\dots,4,
\label{eq:Sij}
\end{equation}
\comm{with the usual notation for the ensemble average (cf. Sec.~\ref{sec:notations}, \eqref{sec:ensemble}.)}
The components \eqref{eq:Sij} give the averaged molecular
direction $\bM_j$ with respect to the laboratory axis $\bL_i$. In general there are 25 independent entries (as expected).
We can alternatively write
\begin{equation}
\bS(\bL_i) = \langle \bM_i \rangle, \,(i = 0,1,\ldots,4)\qquad 
\bS = \sum_{i,j=0}^{4} S_{ij} \bL_{i} \qp \bL_j, \qquad 
\bS = \sum_{k=0}^{4} \langle \bM_k \rangle \qp \bL_k \, .
\end{equation}
In simulations, the orientational probability density $f(\bR)$ is reconstructed by keeping track of the orientations
of a large number $N$ of sample molecules. Thus, the ensemble average in \eqref{eq:Sij} is approximated by the sample 
mean of $\bM_j$ and the components $S_{ij}$ are calculated as
\begin{equation}
S_{ij} = \frac{1}{N}\sum_{\alpha=1}^{N} \bL_{i} \cdot \bM^{(\alpha)}_j ,\qquad i,j= 0,\dots,4,
\end{equation}
where the index $\alpha$ runs over all the molecules in the simulation. 

As a final example, let us consider a system of uniaxial molecules with long axis $\mm_3$. For symmetry reasons,
all the averages of the molecular orientational tensors $\bM_j$ with $j\neq 0$ vanish. The average of $\bM_0$ yields 
the five components ($i = 0,1,\ldots,4$)

\begin{equation}
S_{i, 0} = \sqrt{\frac{3}{2}} \, \bL_{i} \cdot \langle \mm_3 \tp \mm_3 - \tfrac{1}{3}\bI \rangle,
\label{eq:S00}
\end{equation}
which provides a description of the molecular order equivalent to the standard de Gennes $\bQ$ tensor 
\comm{or Saupe ordering matrix \cite{LuckhurstBook,1995deGennes}, as given in Eq.~\eqref{eq:Quniaxial}}.

\subsection{Change of basis and group action}
\label{sec:second-rankOPs_basis}
Let us investigate how the components $S_{ij}$ are affected by a change of the molecular or laboratory frames of reference.
We first consider the components of the ordering tensor with respect to rotated \emph{laboratory axes}. More precisely, 
let $\bl'_i=\bA_{P} \bl_i$ be the unit vectors  along the primed axes, obtained from the old ones by a (proper or improper) 
rotation $\bA_P$. A molecule, whose orientation was described by the rotation $\bR$ with respect to $(\bl_1,\bl_2,\bl_3)$, 
now is oriented as $\bR\bA_{P}\transp$ with respect to the primed frame
\begin{align*}
\bR \bl_i = \mm_i = \bR' \bl'_i = \bR' \bA_P \bl_i \qquad \Rightarrow \qquad \bR'=\bR \bA_{P}\transp .
\end{align*}
Correspondingly, in the five-dimensional space $\V$, the rotation that brings the new frame $\{\bL'_{i}\}$ into coincidence
with the molecular frame $\{\bM_i\}$ is
\begin{equation}
D^{(2)}(\bR\bA_{P}\transp) = D^{(2)}(\bR)D^{(2)}(\bA_{P})\transp, \qquad
D^{(2)}(\bR\bA_{P}\transp)\bL'_i = \bM_i, 
\end{equation}
where we have used the identities $D^{(2)}(\bR)\bL_i = \bM_i$ and $D^{(2)}(\bA_{P})\bL_i = \bL'_i$. The components of the
ordering tensor in the new basis are then calculated as follows
\begin{align}
S'_{i j} & = \bL'_{i} \cdot \langle\bM_{j} \rangle 
= \bL'_{i} \cdot \langle D^{(2)}(\bR\bA_{P}\transp) \rangle \bL'_j
= D^{(2)}(\bA_{P})\bL_{i} \cdot \langle D^{(2)}(\bR) \rangle \bL_j \notag \\
& = \sum_{k=0}^{4} D^{(2)}(\bA_{P})\bL_{i} \cdot (\bL_{k} \qp \bL_{k})\langle D^{(2)}(\bR) \rangle \bL_j 
= \sum_{k=0}^{4} D^{(2)}_{ik}(\bA_{P}\transp) S_{kj},
\label{eq:changebasis}
\end{align}
where $D^{(2)}_{ik}(\bA_{P}\transp) = \bL_{k} \cdot D^{(2)}(\bA_{P})\bL_{i}$. Since only the relative orientation 
of the molecular
frame with respect to the laboratory frame is important, a rigid rotation of \emph{all} the molecules is equivalent
to an inverse
rotation of the laboratory frame. It is easy to check that the ordering tensor in the two cases is the same.
Let $\bA_M$ be a common 
rotation to all the molecular frames, so that the new molecular axes are $\mm'_i = \bA_M \mm_i$. 
The orientation of these axes with
respect to the laboratory frame is given by $\bA_{M}\bR$: $\bA_{M}\bR\bl_i = \bA_{M}\mm_i = \mm'_i$.
The components of the ordering 
tensor then becomes
\begin{align}
S'_{i j} & = \bL_{i} \cdot \langle\bM'_{j} \rangle 
= \bL_{i} \cdot D^{(2)}(\bA_{M})\langle D^{(2)}(\bR) \rangle \bL_j \notag \\
& = \sum_{k=0}^{4} \bL_{i} \cdot D^{(2)}(\bA_{M})(\bL_{k} \qp \bL_{k})\langle D^{(2)}(\bR) \rangle \bL_j
= \sum_{k=0}^{4} D^{(2)}_{ik}(\bA_{M}) S_{kj},
\label{eq:changebasis_mol}
\end{align}
to be compared with \eqref{eq:changebasis}. When combined together, \eqref{eq:changebasis} 
and \eqref{eq:changebasis_mol} show 
that only the relative rotation $\bA_{M}\bA^{T}_{P}$ has a physical meaning.

However, in this context a rotation of the molecular frame has to be interpreted as an orthogonal 
transformation of the molecular 
axes \emph{before} the orientational displacement of the molecule, $\bR$, has taken place. 
In such a case the overall rotation that 
brings the laboratory axes into coincidence with the new molecular axes is described by 
the product $\bR\bA_M$. The new components 
of the ordering tensor are then given by
\begin{align}
S'_{i j} & = \bL_{i} \cdot \langle D^{(2)}(\bR\bA_{M}) \rangle \bL_j 
= \sum_{k=0}^{4} \bL_{i} \cdot \langle D^{(2)}(\bR) \rangle (\bL_{k} \qp \bL_{k}) D^{(2)}(\bA_{M}) \bL_j
= \sum_{k=0}^{4} S_{ik} D^{(2)}_{kj}(\bA_{M}) \, .
\label{eq:changebasis_mol2}
\end{align}

When both laboratory and molecular transformations are allowed, the combination of \eqref{eq:changebasis}
and \eqref{eq:changebasis_mol2} yields
\begin{equation}
S'_{i j} = \sum_{h,k=0}^{4} D^{(2)}_{ih}(\bA_{P}\transp) S_{hk} D^{(2)}_{kj}(\bA_{M}) 
\label{eq:changebasis_tot}
\end{equation}

Dually, we can study the action of two groups $G_P, G_M \subset \Otre$ on $\bS \in \V \otimes \V^*$ by left and right 
multiplication respectively (i.e., $G_P$ acts on the ``phase index'' and $G_M$ on the ``molecular index''). 
According to this \emph{active} interpretation of the orthogonal transformations $\bA_P\in G_P$ and $\bA_M \in G_M$, 
the ordering tensor is transformed in such a way that the following diagram commutes
\begin{center}
\includegraphics[scale=1]{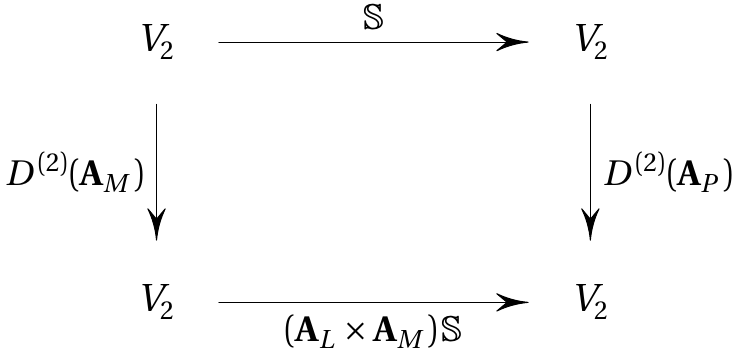} 
\end{center}

\noindent In formula, we have
\begin{align}
(\bA_{P} \times \bA_{M})\,\bS = D^{(2)}(\bA_{P})\, \bS \, D^{(2)}(\bA_{M}\transp) .
\label{eq:activeaction}
\end{align}

\subsection{Molecular and phase symmetry}
\label{sec:second-rankOPs_symmetry}
When dealing with liquid crystals we must distinguish between the symmetry of the molecule and the symmetry of the phase,
shared by the aggregation of the molecules, but not necessarily by the molecules themselves. 
\comm{A thorough description of the symmetries of a  physical system is envisaged 
by the symmetries of the corresponding orientational probability density $f(\bR)$.
However, when we analyse the order of the system only in terms of the descriptor $\bS$, 
some degeneracy arises. The ordering tensor $\bS$ is an averaged quantity obtained by computing 
the second moments of $f(\bR)$. It is therefore possible that systems possessing different physical 
symmetries may be described by the same ordering tensor, since in the averaging procedure some information may be lost.
This means that to be able to distinguish the fine details of these degenerate cases we need to carry on the expansion
of $f(\bR)$ to higher orders. Here, we mainly focus on second-rank properties and we first define what we mean 
by symmetry group of $\bS$.}

According to the action \eqref{eq:activeaction}, \comm{we define the \emph{second-rank molecular symmetry group}
as the set of all elements in $\Otre$ that fix $\bS$ under right multiplication} \footnote{Some authors, 
especially in the mathematical literature, use the term \emph{isotropy group}. Here we prefer to adopt 
the term \emph{symmetry group} \comm{or \emph{stabiliser}} to avoid confusion 
with the term ``isotropy'' as used, e.g., for the isotropic phase, which is a totally different thing.}
\begin{equation}
G_{M}(\bS) = \{\bA_M \in \Otre \suchcol \bS \, D^{(2)}(\bA_{M}\transp) = \bS \},
\end{equation}
and similarly for the \comm{second-rank phase symmetry group, where the multiplication appears on the left}
\begin{equation}
G_{P}(\bS) = \{\bA_P \in \Otre \suchcol D^{(2)}(\bA_{P})\, \bS = \bS \}.
\end{equation}
\comm{We also refer to these subgroups as the right and left stabiliser subgroup for $\bS$. 
A \emph{second-rank symmetry group} or \emph{stabiliser subgroup} is then defined as the subgroup 
of $\Otre \times \Otre$ that collects all the orthogonal transformations, both in the phase and in the molecule,
that leave $\bS$ invariant }(see \cite{ForteVianello} for the corresponding definition in the context of Elasticity Theory).
Mathematically, this is the direct product of $G_P$ and $G_M$: $G(\bS) = G_P(\bS) \times G_M(\bS)$.

The definition of symmetry group explicitly contains the information about the symmetry axes of the molecule and the phase. 
However, our main interest in this Section lies in classifying second-rank ordering tensors with respect 
to their symmetry properties. This means that we wish to introduce, among such tensors, a relation based on the idea
that different materials which can be rotated so that their symmetry groups become identical are `equivalent'.
For instance, two uniaxially aligned liquid crystals are viewed as equivalent in this respect even if the direction 
of alignment of the molecules may be different in the two compounds. Therefore, it is quite natural to think of ordering 
tensors lying on the same $\essotre$--orbit as describing the same material albeit possibly with respect 
to rotated directions.
As a consequence of the definition of the stabilisers $G_{M}(\bS)$ and $G_{P}(\bS)$, the symmetry groups 
with respect to a \emph{rotated} 
frame of reference are simply obtained by conjugation. \comm{For example, if $\bR_{P} \in \essotre$ is
a rotation of the laboratory axes, 
the new ordering tensor is $D^{(2)}(\bR_{P})\,\bS$ and the new symmetry group is conjugated through
$D^{(2)}(\bR_{P})$ to $G_{P}(\bS)$
\begin{align}
G_{P}\big(D^{(2)}(\bR_{P})\,\bS\big) 
& = \{\bA_P \in \Otre \suchcol D^{(2)}(\bA_{P})D^{(2)}(\bR_{P})\, \bS = D^{(2)}(\bR_{P})\,\bS \} \notag \\
& = \{\bA_P \in \Otre \suchcol D^{(2)}(\bR_{P})\transp D^{(2)}(\bA_{P}) D^{(2)}(\bR_{P})\, \bS = \bS \} \notag \\
& = D^{(2)}(\bR_{P}) G_{P}(\bS) D^{(2)}(\bR_{P})\transp.
\end{align}
}
We regard two ordering tensors $\bS_1$ and $\bS_2$ which are related by a rotation of the axes as representing 
the same material 
and hence equivalent. Thus, we speak about (second-rank) \emph{symmetry classes} and say that the two ordering 
tensors belong 
to the same symmetry class (and are therefore equivalent) when their \comm{stabiliser subgroups} are conjugate. 
More precisely,
we write $\bS_1 \sim \bS_2$ if and only if there exist two rotations $\bR_{P}, \bR_{M} \in \essotre$ such that 
\begin{equation}
D^{(2)}(\bR_{P}) G_{P}(\bS_1) D^{(2)}(\bR_{P})\transp = G_{P}(\bS_2) \qquad \text{ and } \qquad
D^{(2)}(\bR_{M}) G_{M}(\bS_1) D^{(2)}(\bR_{M})\transp = G_{M}(\bS_2).
\end{equation}

\comm{Finally, we say that the point groups $G_1$ and $G_2$ are (second-rank) \emph{indistinguishable symmetries 
for the physical system} if, for any two probability densities $f_1$ and $f_2$ that are fixed by $G_1$ and
$G_2$\footnote{Since $G_1,G_2 \subset \Otre$, a correct definition would require $f_1,f_2:\Otre \to \R_+$. 
This point is discussed in the Appendix, to avoid diverting here from the main discourse.}, respectively, 
the corresponding ordering tensors $\bS_1$ and $\bS_2$ belong to the same symmetry class ($\bS_1 \sim \bS_2$). }

\subsection{Symmetry classes of second-rank ordering tensors}
\label{sec:second-rankOPs_symmetryclasses}
A preliminary problem which we need to address is counting and determining all symmetry classes for second-rank 
ordering tensors.
An analogous problem in Elasticity, i.e., determining the symmetry classes of the linear elasticity tensor,
is discussed in
\cite{ForteVianello}. However, as we shall see, in our case this determination is simpler because we are dealing with
\emph{irreducible second-rank} tensors (instead of reducible fourth-rank) and the action of molecular and phase symmetry 
can be studied separately.

It is worth remarking that there is not a one to one correspondence between the \comm{stabiliser subgroups 
for $\bS$ and the point groups in three dimensions: liquid crystal compounds possessing different (molecular or phase) 
physical symmetry may have the same stabiliser for $\bS$ and thus may belong to the same second-rank symmetry class.}
This fact is related to the truncation of the probability density used to define the order parameters: 
at the second-rank level the ordering tensors may coincide even if the actual material symmetry is different 
as in the truncation process some information about the molecular distribution is lost. For example, at the 
second rank level, materials with $C_{3h}$ and $D_{\infty h}$ symmetry are effectively indistinguishable. 
This would not be true if we considered third-rank tensorial properties.
However, for the sake of simplicity and also because it is most widely adopted in the literature, we will 
consider only second-rank order parameters. 

A number of authors have made the same classification based on the number of non-vanishing independent order parameters
for each group \cite{1985zannoni,2001luckhurst,LuckhurstBook,2006Mettout}. However, in our view, this classification 
of the symmetry classes rests on two standard theorems, which we now state without proof. The first theorem is known 
as Hermann-Herman' theorem in Crystallography. The interested reader can consult the original 
references~\cite{1934Hermann,1945HermanB}; Refs. \cite{2004Slawinski,2001Wadhawan} for a proof and
Refs.~\cite{1987Wadhawan,1982Sirotin,1998Handbook} for a more accessible account of this result.

\begin{theorem}[Hermann-Herman] Let $\,\bTT\,$ be an $r$-rank ($r>0$) tensor in $W^{\tp r}$, where $W$ is a 
\emph{3-dimensional real}
vector space. If $\,\bTT\,$ is invariant with respect to the group $C_n$ of $n$-fold rotations about a
fixed axis and $n>r$, 
then it is \comm{$C_{\infty}$}-invariant relative to this axis (i.e., it is $C_{m}$-invariant for all $m \geq n$).
\label{thm:Herma}
\end{theorem}

Quoting Herman, from \cite{2004Slawinski,1945HermanB} ``If the medium has a rotation axis of symmetry $C_n$ of order $n$,
it is axially isotropic relative to this axis for all the physical properties defined by the tensors of the rank 
$r=0,1,2,\ldots,(n-1)$.''

\comm{
To perform the classification of the second-rank symmetry classes, we also need to recall 
a standard classification theorem in Group Theory \cite{ForteVianello,Sternberg,1972Bredon}.
\begin{theorem} Every closed subgroup of $\essotre$ is isomorphic to exactly one of the following groups ($n\geq 2$): 
$C_1$, $C_n$, $D_n$, $T$, $O$, $I$, $C_{\infty}$, $C_{\infty v}$, $\essotre$.
\label{thm:closedsubgroups}
\end{theorem}
}
In view of these theorems, for $j=2$ we obtain a result that allows collecting all point groups in five classes, 
which greatly simplifies the classification of phase or molecular symmetries.

\comm{\begin{proposition} There are exactly five (phase or molecular) symmetry-classes of the second-rank ordering
tensor:
\emph{Isotropic, Uniaxial (Transverse Isotropic), Orthorhombic, Monoclinic and Triclinic}. The corresponding stabiliser
subgroups for $\bS$ are: $O(3)$, $D_{\infty h}$, $D_{2 h}$, $C_{2 h}$ and $C_i$. The last column in the table collects the second-rank indistinguishable symmetries, i.e., 
physical symmetries that yield equivalent second-rank ordering tensors.
\begin{center}
\includegraphics[scale=0.8]{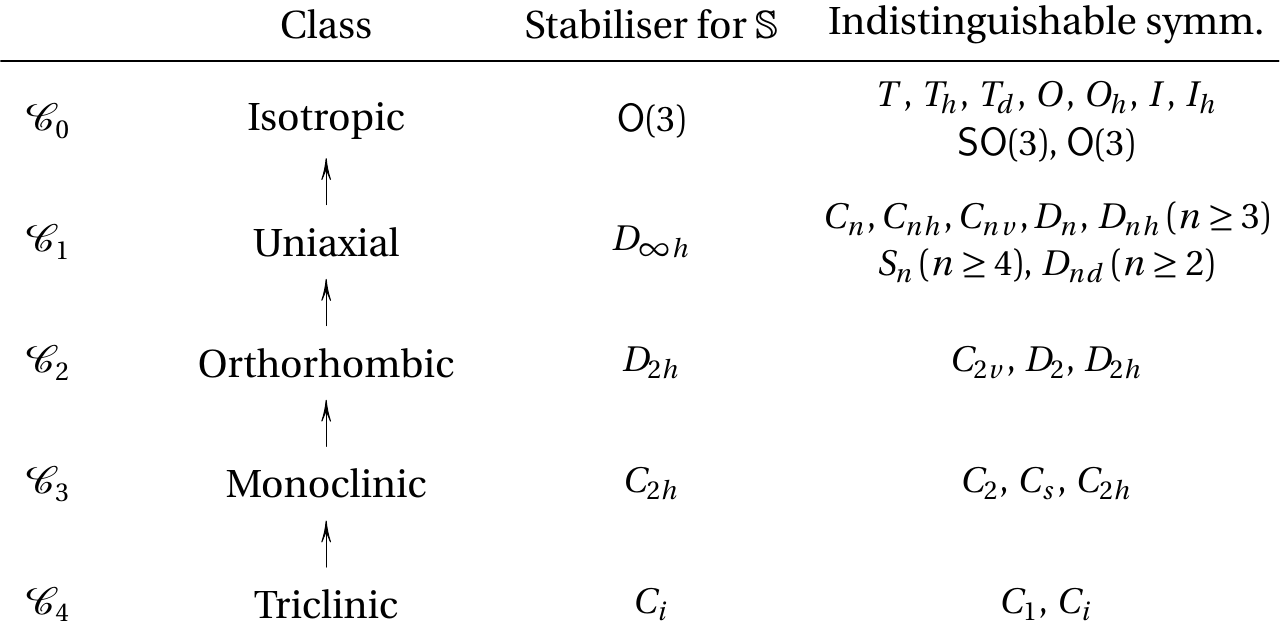}
\end{center}
Furthermore, since $\bS$ vanishes in the isotropic class, there are only $4 \times 4 + 1 = 17$ possible 
different combinations of molecular and phase symmetries that can be distinguished at the level of second-rank
order parameters.
\end{proposition}}
\begin{proof} The proof is a consequence of the following remarks.
\begin{enumerate} 
\item[(1)] Since the symmetry group $G$ is the direct product of $G_M$ and $G_P$, the action of a symmetry 
transformation can be studied independently for the molecules and the phase. This greatly simplify the classification 
(by contrast, this fact is not true in Elasticity).
\item[(2)] The definition of $D^{(2)}(\bR)$, as given in \eqref{eq:conjugation}, involves exactly twice the product of the rotation matrix
$\bR$. This implies, by Herman's theorem, that all the groups with an $n$-fold rotation axis with $n>2$ are effectively
indistinguishable from \comm{the $C_{\infty}$-symmetry} about that axis.  

\item[(3)] Furthermore, \comm{Eq.\eqref{eq:conjugation} shows immediately that $D^{(2)}(\iota\bR) = D^{(2)}(\bR)$. 
In particular, this yields $D^{(2)}(\iota) = D^{(2)}(\bI)$, $D^{(2)}(\sigma_h) = D^{(2)}(C_{2z})$,
and $D^{(2)}(\sigma_v) = D^{(2)}(C_{2x})$ so that the inversion and the identity are represented by the same matrix; 
a horizontal mirror reflection is equivalent to a 2-fold rotation about the main axis $z$, and a vertical mirror reflection
is equivalent to a 2-fold rotation about an orthogonal axis $x$. Therefore, the classification can be first performed on 
the subgroups of $\essotre$ since each class will have at least one representative subgroup in $\essotre$. 
The indistinguishable subgroups of $\Otre$ are then classified by considering the trivial actions of $\iota$ or $\sigma_h$.}
\item[(4)] \comm{Whenever a point group has two independent rotation axes of order $n>2$, it must contain 
two distinct copies of $C_{\infty}$}. By checking the list of the possible closed subgroups of $\essotre$ 
(theorem~\ref{thm:closedsubgroups}) we see that 
it must be the whole $\essotre$.
\item[(5)] \comm{Finally, the stabiliser of $\bS$ in each class is determined by the largest group in the class.}

\comm{The proof is then completed by inspection of the various point groups. For example, point (4) 
immediately shows that the higher order groups $T$, $O$ and $I$ are all indistinguishable from $\essotre$ and 
thus they all belong to the same class. Then, from point (3) we learn that adding an inversion or a mirror
reflection has no effect on $\bS$. Hence we can also classify $I_h$  $T_h$, $T_d$, $O_h$, and $\Otre$ as belonging
to the same class (called \emph{isotropic}). The stabiliser for $\bS$ is the largest among such groups and it is $\Otre$.

Likewise, axial groups with a 3-fold or higher rotation axis are to be placed in the same class of $C_{\infty}
=\essodue$, according to the remark in point (2).} \comm{We now choose a frame of reference with the rotation 
axis along the $z$ coordinate and use the basis $\{\bL_i\}$, as given in Eq.\eqref{eq:molorienttensors}, 
to represent $\bS$. The only basis tensor that is fixed by an arbitrary rotation about the $z$-axis is $\bL_0$.
Therefore, in this basis a $C_{\infty}$-invariant $\bS$ must be written either as 
\[
\bS = \sum_{i=0}^{4} S_{i,0} \bL_{i} \kp \bL_{0}, \qquad \text{ or }  \qquad \bS = \sum_{j=0}^{4} S_{0,j} 
\bL_{0} \kp \bL_{j} ,                                                                                                                                                                                                                                                                                                                                                                                                                                                                                                                                                                                                                                                                                                                                                                                                                                                                                                                                 
                                                                     \]
depending on whether $C_{\infty}$ is acting on the right or the left. The only non-vanishing 
entries are either the first column or the first row of the matrix representing $\bS$ (see also Table \ref{tab:canonical}).
Since the tensor $\bL_0$ is not affected by a mirror reflection across planes through the $z$-axis, or by a $C_2$ rotation
about the $x$-axis, the ordering tensor $\bS$ is (right- or left-) fixed also by $\sigma_v$ and $C_{2x}$. This brings
$C_{\infty v}=\Odue$ and $D_{\infty h}$ into the class. From these it is then easy to classify all the other groups 
in the \emph{uniaxial} class.

By contrast, the \emph{orthorhombic} class and the \emph{monoclinic} class are only composed by groups
with 2-fold rotation axes, so that Theorem \ref{thm:Herma} does not apply. The two classes are distinguished 
by the presence or absence of a second rotation axis orthogonal to $z$ and point (3) allows us to identify 
the indistinguishable subgroups in each class. Finally, the \emph{triclinic} class collects the remaining 
trivial groups $C_1$ and $C_i$. All the symmetry classes are disjoint since it is possible to provide independent
examples of an ordering tensor in each class.}
\end{enumerate} 
\end{proof}

\section{Identification of the nearest symmetric ordering tensor}
\label{sec:identification}

\subsection{Invariant projection}
\label{sec:identification_projection}
\comm{If $G_M, G_P \subset \Otre$ represent the real symmetry of the material, the order parameter $\bS$
must be an invariant tensor for the action of $G_P \times G_M$. However, in practice, $\bS$ will not 
be fixed exactly by any non-trivial group, due to measurements errors or non perfect symmetry of the real system. 
Therefore, our problem can be stated as follows: given a measured 5 $\times$ 5 
ordering tensor $\bS$ and assuming two specific stabiliser subgroups $G_M, G_P$, find the ordering tensor $\bSs$
which is 
\mbox{$(G_P \times G_M)$--invariant} and is closest to $\bS$. We now discuss two related sub-problems, namely,
(1) how to find $\bSs$ and (2) what is meant by ``closest to $\bS$''.}

\comm{Let us define the fixed point subspace $L(\V)^G$ as the space of the tensors $\bS$ that are fixed under $G$:
\begin{equation}
L(\V)^G = \{\bS \in L(\V): g \bS=\bS, \forall g \in G\}.
\end{equation}
}
The invariant projection onto \comm{$L(\V)^{G_P \times G_M}$} can be easily obtained by averaging over the group
\begin{equation}
\bSs = \frac{1}{|G_P| |G_M|}\sum_{\bA_{P} \in G_P} \sum_{\bA_{M} \in G_M} D^{(2)}(\bA_{P})\, \bS \,D^{(2)}(\bA_{M}\transp) , 
\label{eq:Reynolds}
\end{equation}
where $|G_P|$ and $|G_M|$ are the orders of the two (finite) groups. We observe that the transposition of $\bA_M$ in
Eq.~\eqref{eq:Reynolds} is unnecessary since we are summing over the whole group, but it is maintained here for consistency
with \eqref{eq:activeaction}. This averaging procedure is standard in many contexts and takes different names accordingly.
It is called ``averaging over the group'' in Physics, ``projection on the identity representation'' in Group Theory and
``Reynolds operator'' in Commutative Algebra. The expression \eqref{eq:Reynolds} is manifestly invariant by construction.
It is also easy to show that it constitutes an \emph{orthogonal projection}\comm{, with respect to the standard 
Frobenius inner product, i.e., the natural extension of Eq.\eqref{eq:dot_tensors} to $L(\V)$}. Hence, the distance 
of $\bSs$ from the original
$\bS$ is minimal and can be easily computed, as we now discuss.

\begin{lemma} 
The Reynolds operator as given in \eqref{eq:Reynolds}, with $G_M, G_P \subset \Otre$, is an orthogonal projector
\comm{onto $L(\V)^{G_P \times G_M}$}
\end{lemma}
\begin{proof} Let us introduce the linear operator $\Rey$ such that $\Rey(\bS) = \bSs$ as given in Eq.~\eqref{eq:Reynolds}.
For convenience, we will use the following more compact notation for the Reynolds operator:
\[\Rey (\bS) = \frac{1}{|G|} \sum_{g\in G} g\bS , \]
where $G$ is the direct product of groups $G=G_P \times G_M$ and $g= D^{(2)}(\bA_{P})\tp \,D^{(2)}(\bA_{M})$ 
is the tensor (Kronecker) product of the matrix representation. First, we observe that by construction $\Rey^2 = \Rey$ 
as we are summing over the whole group.

Next, we show that $\Rey\transp = \Rey$. Since $G_M, G_P \subset \Otre$, it follows from \eqref{eq:D(R)orthogonality}
that $g^{-1} = g\transp$. Therefore, for any two ordering tensors $\bS, \bTT \in L(\V)$, we have
\begin{align}
\Rey(\bTT) \cdot \bS 
= \frac{1}{|G|} \sum_{g\in G} g\bTT \cdot \bS 
= \frac{1}{|G|} \sum_{g\in G} \bTT \cdot g\transp\bS
=\frac{1}{|G|} \sum_{g\in G} \bTT \cdot g^{-1}\bS 
= \frac{1}{|G|} \sum_{g\in G} \bTT \cdot g\bS 
= \bTT \cdot \Rey(\bS),
\end{align}
where we have used the fact that summing over $g^{-1}$ is the same as summing over $g$, since $G$ is a group.
Finally, we define $\bS^{\perp}:= \bS-\Rey(\bS)$ and obtain
\begin{align}
\Rey(\bS) \cdot \bS^{\perp} = \Rey(\bS) \cdot \big(\bS-\Rey(\bS)\big) 
= \Rey(\bS) \cdot \bS- \Rey(\bS) \cdot \Rey(\bS) 
= \Rey(\bS) \cdot \bS- \Rey^2(\bS) \cdot \bS =0 .
\end{align}
\end{proof}

The above lemma suggests that the \comm{Frobenius} norm $\|\bS^{\perp}\|$ is a suitable candidate 
for the ``distance from the $G$-invariant subspace''. It is worth noticing explicitly that this distance is
properly defined as its calculation does not depend 
on the particular chosen matrix representation of $\bS$, i.e., on the molecular and laboratory axes.
Furthermore, since by orthogonality we have $\|\bS\|^2 = \|\bSs\|^2 + \|\bS^{\perp}\|^2$, 
we readily obtain the following expression for the distance
\begin{equation}
\mathrm{d}(\bS, \comm{L(\V)^{G_P \times G_M}}) = \|\bS^{\perp}\| = \sqrt{\phantom{\big(}\|\bS\|^2 - \|\bSs\|^2} .
\label{eq:distance}
\end{equation}

\subsection{Canonical matrix representation}
\label{sec:identification_canonical}
The matrix representation of the ordering tensor $\bS$ assumes a particular simple form when the molecule 
and the laboratory axes are chosen in accordance with the molecular and phase symmetry, respectively.
When the $z$-axis is assumed to be the main axis of symmetry and the basis tensors $\bM_i$ and $\bL_{j}$ 
are chosen accordingly, the projection \eqref{eq:Reynolds} gives a ``canonical'' form for the ordering tensor, 
shown explicitly in Table~\ref{tab:canonical}. The number of non-vanishing entries in each matrix corresponds
to the number of independent (second-rank) order parameters necessary to describe the orientational order of 
molecules of the given symmetry in a given phase. The same results can be obtained by direct computation using 
the invariance of each matrix entry by symmetry transformations (see for example the calculations at the end
of \cite{07Rosso} relative to a nematic biaxial liquid crystal).

\newcommand{\zz}{\cdot}
%\rotatebox{90}{
\begin{table}[h]
\begin{scriptsize}
{\renewcommand{\arraystretch}{1.1}  \tabcolsep=1pt
\begin{tabular}{|c|cccc|}
\hline
\backslashbox{Mol.}{Phase}
& Triclinic & Monoclinic & Orthorhombic & Uniaxial \\[2mm] \hline
&&&&\\[-1mm]
% ***********************************
% RIGA 1
Triclinic
&
$\begin{pmatrix}
S_{00} & S_{01} & S_{02} & S_{03} & S_{04} \\
S_{10} & S_{11} & S_{12} & S_{13} & S_{14} \\
S_{20} & S_{21} & S_{22} & S_{23} & S_{24} \\
S_{30} & S_{31} & S_{32} & S_{33} & S_{34} \\
S_{40} & S_{41} & S_{42} & S_{43} & S_{44}
\end{pmatrix} $
&
$\begin{pmatrix}
S_{00} & S_{01} & S_{02} & S_{03} & S_{04} \\
S_{10} & S_{11} & S_{12} & S_{13} & S_{14} \\
S_{20} & S_{21} & S_{22} & S_{23} & S_{24} \\
\zz & \zz & \zz & \zz & \zz \\
\zz & \zz & \zz & \zz & \zz
\end{pmatrix} $
& 
$\begin{pmatrix}
S_{00} & S_{01} & S_{02} & S_{03} & S_{04} \\
S_{10} & S_{11} & S_{12} & S_{13} & S_{14} \\
\zz & \zz & \zz & \zz & \zz \\
\zz & \zz & \zz & \zz & \zz \\
\zz & \zz & \zz & \zz & \zz
\end{pmatrix} $
&
$\begin{pmatrix}
S_{00} & S_{01} & S_{02} & S_{03} & S_{04} \\
\zz & \zz & \zz & \zz & \zz \\
\zz & \zz & \zz & \zz & \zz \\
\zz & \zz & \zz & \zz & \zz \\
\zz & \zz & \zz & \zz & \zz
\end{pmatrix} $ \\[8mm]
% ***********************************
% RIGA 2
Monoclinic
& %\cellcolor[gray]{0.9}
$\begin{pmatrix}
S_{00} & S_{01} & S_{02} & \zz & \zz \\
S_{10} & S_{11} & S_{12} & \zz & \zz \\
S_{20} & S_{21} & S_{22} & \zz & \zz \\
S_{30} & S_{31} & S_{32} & \zz & \zz \\
S_{40} & S_{41} & S_{42} & \zz & \zz
\end{pmatrix} $
&
$\begin{pmatrix}
S_{00} & S_{01} & S_{02} & \zz & \zz \\
S_{10} & S_{11} & S_{12} & \zz & \zz \\
S_{20} & S_{21} & S_{22} & \zz & \zz \\
\zz & \zz & \zz & \zz & \zz \\
\zz & \zz & \zz & \zz & \zz
\end{pmatrix} $
& 
$\begin{pmatrix}
S_{00} & S_{01} & S_{02} & \zz & \zz \\
S_{10} & S_{11} & S_{12} & \zz & \zz \\
\zz & \zz & \zz & \zz & \zz \\
\zz & \zz & \zz & \zz & \zz \\
\zz & \zz & \zz & \zz & \zz
\end{pmatrix} $
&
$\begin{pmatrix}
S_{00} & S_{01} & S_{02} & \zz & \zz \\
\zz & \zz & \zz & \zz & \zz \\
\zz & \zz & \zz & \zz & \zz \\
\zz & \zz & \zz & \zz & \zz \\
\zz & \zz & \zz & \zz & \zz
\end{pmatrix} $ \\[8mm]
% ***********************************
% RIGA 3
Orthorhombic
& %\cellcolor[gray]{0.9}
$\begin{pmatrix}
S_{00} & S_{01} & \zz & \zz & \zz \\
S_{10} & S_{11} & \zz & \zz & \zz \\
S_{20} & S_{21} & \zz & \zz & \zz \\
S_{30} & S_{31} & \zz & \zz & \zz \\
S_{40} & S_{41} & \zz & \zz & \zz
\end{pmatrix} $
& %\cellcolor[gray]{0.9}
$\begin{pmatrix}
S_{00} & S_{01} & \zz & \zz & \zz \\
S_{10} & S_{11} & \zz & \zz & \zz \\
S_{20} & S_{21} & \zz & \zz & \zz \\
\zz & \zz & \zz & \zz & \zz \\
\zz & \zz & \zz & \zz & \zz
\end{pmatrix} $
& 
$\begin{pmatrix}
S_{00} & S_{01} & \zz & \zz & \zz \\
S_{10} & S_{11} & \zz & \zz & \zz \\
\zz & \zz & \zz & \zz & \zz \\
\zz & \zz & \zz & \zz & \zz \\
\zz & \zz & \zz & \zz & \zz
\end{pmatrix} $
&
$\begin{pmatrix}
S_{00} & S_{01} & \zz & \zz & \zz \\
\zz & \zz & \zz & \zz & \zz \\
\zz & \zz & \zz & \zz & \zz \\
\zz & \zz & \zz & \zz & \zz \\
\zz & \zz & \zz & \zz & \zz
\end{pmatrix} $ \\[8mm]
% ***********************************
% RIGA 4
Uniaxial
& %\cellcolor[gray]{0.9}
$\begin{pmatrix}
S_{00} & \zz & \zz & \zz & \zz \\
S_{10} & \zz & \zz & \zz & \zz \\
S_{20} & \zz & \zz & \zz & \zz \\
S_{30} & \zz & \zz & \zz & \zz \\
S_{40} & \zz & \zz & \zz & \zz
\end{pmatrix} $
& %\cellcolor[gray]{0.9}
$\begin{pmatrix}
S_{00} & \zz & \zz & \zz & \zz \\
S_{10} & \zz & \zz & \zz & \zz \\
S_{20} & \zz & \zz & \zz & \zz \\
\zz & \zz & \zz & \zz & \zz \\
\zz & \zz & \zz & \zz & \zz
\end{pmatrix} $
& %\cellcolor[gray]{0.9}
$\begin{pmatrix}
S_{00} & \zz & \zz & \zz & \zz \\
S_{10} & \zz & \zz & \zz & \zz \\
\zz & \zz & \zz & \zz & \zz \\
\zz & \zz & \zz & \zz & \zz \\
\zz & \zz & \zz & \zz & \zz
\end{pmatrix} $
&
$\begin{pmatrix}
S_{00} & \zz & \zz & \zz & \zz \\
\zz & \zz & \zz & \zz & \zz \\
\zz & \zz & \zz & \zz & \zz \\
\zz & \zz & \zz & \zz & \zz \\
\zz & \zz & \zz & \zz & \zz
\end{pmatrix} $ \\[8mm] \hline
\end{tabular}}
\end{scriptsize}
%}
\caption{\small{Canonical form \comm{of $\bSs$ obtained by the projection of a generic ordering tensor $\bS$ 
using Eq.\eqref{eq:Reynolds}}. The \comm{basis tensors} $\bM_i$ and $\bL_{j}$, as given in \eqref{eq:molorienttensors},
are adapted to the molecular 
and phase symmetries, where the $z$-axis is the main axis of symmetry. For the sake of readability, 
dots stand for vanishing entries. The isotropic class is omitted since the entries of the corresponding ordering 
tensors all vanish.}}
\label{tab:canonical}
\end{table}

We immediately read from Eq.~\eqref{eq:S00} that 
\begin{equation}
S_{00} = \bL_{0}\cdot \langle \bM_{0} \rangle 
= \bl_{3}\cdot \left(\frac{3}{2}\left\langle \mm_3 \tp \mm_3 - \frac{1}{3}\bI \right\rangle \right) \bl_3
= \frac{1}{2}\left\langle 3(\bl_{3}\cdot\mm_3)^2 - 1 \right\rangle = S,
\end{equation}
where $S$ is the standard uniaxial ($D_{\infty h}$) order parameter (degree of orientation).
This is due to the fact that $\mm_3$ is the main rotation axis of the molecule and we have chosen the laboratory axis $\bl_3$ 
along the uniaxial symmetry axis of the phase (see Sec. \ref{sec:background}). Likewise, we can see that when the frames of reference are adapted
to the molecular and phase symmetries of the system, $S_{10}$ corresponds to the degree of phase biaxiality $P$ while
$S_{01}$ and $S_{11}$ are the two additional nematic biaxial ($D_{2h}$) order parameters, usually written as $D$ and $C$ 
(see~\cite{07Rosso} for notations)\footnote{NB: In our notation $S_{ij}=\bL_i \cdot \langle \bM_j \rangle$ (see \eqref{eq:Sij}), whilst in~\cite{07Rosso} $S_{ij}=\langle \bM_i \rangle \cdot \bL_j$.}.

However, some ambiguity arises from this definition because in low symmetry groups not all the axes are uniquely defined
by the group operations. This, in a sense, suggests that, contrary to common understanding, the matrix entries of $\bS$ 
are not suitable candidates for the order parameters, but rather the ordering tensor $\bS$ as a whole should be considered 
as the correct descriptor of the molecular order in a low symmetry system. For example, while the $D_{2h}$ symmetry uniquely 
identifies three orthogonal directions that can be used as coordinate axes, the $C_{2h}$ symmetry only identifies
the rotation axis $z$, but gives no indication on how to identify the coordinate axes $x$ and $y$. We note that this choice 
is important, because some matrix entries are altered when we choose different $x$ and $y$ axes. By contrast, in a system
with $D_{\infty h}$ symmetry the choice of the $x$ and $y$ axes is not important as all the associated order parameters 
vanish by symmetry. An extreme example is furnished by the triclinic case, where the molecule 
and the phase posses no symmetry.
Here, there is no reason to select one preferred direction with respect to any other and any coordinate frame should 
be equivalent. However, the matrix entries of $\bS$ surely depend on the chosen axis.

Therefore, in general, a different choice of the $x$ and $y$ axes could lead to different values for the matrix entries
and the canonical form of the matrix $\bS$ is thus not uniquely defined. By contrast, the \emph{structure} of the matrix 
as given in Table~\ref{tab:canonical}, and in particular the position of the vanishing entries of $\bS$, are not affected 
by such a change of basis. This suggests that the order parameters should not be 
\comm{identified with the matrix entries, but rather with the \emph{invariants} of $\bS$ (for example, its eigenvalues).}

Finally, it is important to observe that the molecular and phase symmetries are uniquely identified once 
we recognise the structure of the canonical form, independently of the definition of the order parameters.

\begin{proposition} An ordering tensor $\,\bS\,$ \comm{is fixed by any group in} the symmetry class $\Cl_{(p,m)} =
\Cl_{p}\times \Cl_{m}$ if and only if there exist molecular and  laboratory frames of reference such that the matrix
representation of $\,\,\bS\,$ 
has a canonical form given by the entry $(p,m)$ of Table~\ref{tab:canonical}. 
\end{proposition}
\begin{proof} This can be checked by direct computation. The canonical forms in Table~\ref{tab:canonical} are obtained
by projection of a full matrix $\bS$ onto the \comm{fixed point subspace $L(\V)^{G_P \times G_M}$ using 
the Reynolds operator \eqref{eq:Reynolds}, where $G_P$ and $G_M$ are any two of the indistinguishable subgroups 
in the given symmetry class.}
The molecular and laboratory $z$ axes are assumed to be the main axes of symmetry. Let $\Rey_{(p,m)}$ be 
the \comm{corresponding} Reynolds operator. \\
If $\bS$ \comm{is fixed by $G_P \times G_M$}, then $\Rey_{(p,m)}(\bS)=\bS$. By construction, $\bS$ must
be equal to the canonical form, once we choose the molecular frame and the laboratory frame \comm{in accordance 
with the symmetry axes of $G_P$ and $G_M$.}\\
Vice versa, if $\bS$ is in a canonical form $(p,m)$, then we can check directly that it is invariant
for the application of
$\Rey_{(p,m)}$ for a suitable choice of the symmetry axes.
\end{proof}

\subsection{Determining the effective phase}
\label{sec:identification_effectivephase}
In order to determine the phase \comm{of the system} it is not necessary to identify the scalar order parameters suitable
to describe the molecular order in each case. As discussed in the previous section, it is sufficient 
to study whether the ordering tensor $\bS$ is close to one of the canonical forms of Table~\ref{tab:canonical},
for a suitable choice of the molecular and laboratory axis. 

To evaluate the distance of $\bS$ from a given \comm{fixed point subspace}, we choose any group
in the \comm{corresponding class} and describe its elements concretely by assigning the orthogonal
transformations and assuming the rotation axes. However, the rotation axes are not in general known 
(and are indeed part of the sought solution). Therefore, we calculate the distance with respect to all of the possible directions of the symmetry axes and define the distance from the symmetry class as the minimum among the distances.

It is useful to define the \emph{coefficient of discrepancy}~\cite{1998Geymonat} as the minimum relative 
distance of $\bS$ from a symmetry class

\begin{equation}
c(p,m) = \min \left\{\frac{d(\bS,\comm{L(\V)^{G}})}{\| \bS \|}\, : \, G = G_P \times G_M \in \Cl_{(p,m)} \right\}.
\label{eq:discrepancy_2}
\end{equation}
When the molecular symmetry is known, say $G_M$, we perform the optimisation only with respect to the phase 
groups and the coefficient of discrepancy depends only on the phase-index
\begin{equation}
c(p) = \min \left\{\frac{d(\bS,\comm{L(\V)^{G\times G_M}})}{\| \bS \|}\, : \, G = G_P \in \Cl_{p} \right\}.
\label{eq:discrepancy}
\end{equation}

The algorithm we propose is described schematically as follows (we specifically concentrate 
on the phase symmetry as the molecular symmetry can usually be assumed to be known \emph{a-priori}).
\comm{
\begin{enumerate}
\item[(1)] If $\| \bS \|=0$, then the phase is isotropic.
\item[(2)] If $\| \bS \| \neq 0$, then \textbf{loop over} the phase symmetry classes for \bm{$p=1,2,3$} 
(there is no need to minimise in the trivial class $\Cl_4$).

\setlength{\itemindent}{0.6cm}
\item[(2.1)] \textbf{Choose the lowest order group in $\Cl_{p}$}. Select an abstract group
$G_P$ in the chosen class. This is reasonably the one with lowest order.

\item[(2.2)] \textbf{Minimisation step.} The distance~\eqref{eq:distance} depends,
via the Reynolds operator~\eqref{eq:Reynolds}, on the concrete realisation of $G_P$, i.e.,
on the direction of the symmetry axes. Therefore we need to minimise the distance \eqref{eq:distance} 
with respect to all possible directions of the rotation axes allowed by the specific abstract group. 

\setlength{\itemindent}{0cm}
\item[(3)] \textbf{Selection step.} Clearly, the correct symmetry class is not that of the lowest distance.
For instance, any ordering tensor $\bS$ has a vanishing distance with respect to the triclinic class 
(absence of symmetry). As a second example, 
since the lattice of the stabiliser subgroups is composed by a single chain, an uniaxial ordering 
tensor has zero distance also from all the previous classes in the chain, i.e. orthorhombic, monoclinic and triclinic. \\
In principle, when $\bS$ is free from numerical errors, we should choose the class with the highest 
symmetry and vanishing coefficient of discrepancy. However, in practice we will choose the highest symmetry compatible with a coefficient of discrepancy \eqref{eq:discrepancy} not exceeding the experimental error (or simulation error).
\end{enumerate}
}

\section{Examples}
\label{sec:examples}
We now briefly describe how to apply our algorithm to determine the phase symmetry of two liquid crystal compounds, 
one composed of uniaxial ($D_{\infty h}$) molecules and the other made of biaxial ($D_{2h}$) molecules. These symmetry
assumptions are quite standard and shared by most theoretical studies in the field 
(see for example~\cite{2011turziC2h,2012turziD2h,2013TurziSluckin,2003Virga,universal,bisi2011,bisi2013}).
In the following we assume (quite reasonably) 
that the molecular symmetry and the molecular axes are known \emph{a-priori}, as is the case of Monte Carlo simulations 
or mean-field analysis. Therefore, no minimisation is required with respect to the molecular frame of reference.
Rather, we concentrate on the determination of the phase symmetry and its principal axes. 

In the examples that follow, we have produced possible outcomes of simulations by computing $\bS$ 
for a system composed of a large number of molecules. We have randomly perturbed their initial perfect uniaxial 
or biaxial order to simulate more realistic results, affected by noise.

\subsection{Uniaxial phase}
\label{sec:examples_uniaxial}
First, let us consider a uniaxial phase with symmetry axis determined by the Euler angles
$\alpha = 60^\circ$, $\beta = 30^\circ$, $\gamma = 0$.
After introducing a random perturbation, the degree of order is $S=S_{00}\approx 0.69$,
and the following ordering matrix is obtained:
\begin{equation}
\bS 
= \begin{pmatrix}
0.404 & -0.016 & 0.005 & -0.013 & -0.005 \\
-0.090 & -0.008 & -0.005 & 0.016 & -0.009 \\
0.122 & 0.006 & -0.009 & -0.004 & -0.012 \\
0.476 & 0.009 & -0.005 & 0.000 & 0.020 \\
0.234 & -0.007 & -0.014 & -0.011 & 0.010
\end{pmatrix},
\label{eq:bS_example_uniaxial}
\end{equation}
written with respect to an arbitrarily chosen laboratory frame, but with the molecular frame accurately 
selected according to the molecular symmetry. A quick look at Table~\ref{tab:canonical} correctly suggests that $\bS$ 
refers to a system of uniaxial molecules, although affected by experimental or numerical errors (only the first column 
contains significantly non-vanishing entries). However, the symmetry class of the phase, the axes of symmetry and the 
relevant order parameter(s) are yet unknown. To this end, we project $\bS$ on each symmetry class and compute
the coefficient of discrepancy \eqref{eq:discrepancy} in each case. It is unnecessary to project on the triclinic class 
since by definition this includes all the possible ordering tensors and the 
distance \eqref{eq:distance} is therefore always zero.

The optimal choice for the symmetry axes is given by minimising the distance of $\bS$ from each symmetry class.
We implement this minimisation procedure rather naively by uniformly sampling $\essotre$, i.e., we generate $N=10^4$ 
orientations of the laboratory axes uniformly. There are of course more refined optimisation algorithms that could yield 
far better results, but these fall outside the scope of the present paper. We intend to study this computational 
issue more deeply in a subsequent paper. The result of our analysis is summarised in the following table
\begin{center}
\begin{tabular}{c|ccccc}
$p$ & $0$ & $1$ & $2$ & $3$ & $4$ \\[1mm] \hline
\vspace{1mm}$c(p)$ & 1 & 0.063 & 0.056 & 0.049 & 0
\end{tabular}
\end{center}

The coefficients of discrepancy suggest that the phase is uniaxial and $\bS \in \Cl_1$. The Euler angles
of the phase axes are then found to be $\alpha \approx 63.9^\circ$, $\beta \approx 31.1^\circ$, $\gamma \approx 124^\circ$.
Note that the value of the proper rotation angle, $\gamma$, is irrelevant in a uniaxial phase. Finally,
we can write the ordering tensor in the symmetry adapted frame of reference according to \eqref{eq:changebasis_tot}.
The new ordering matrix $\bS'$ reads
\begin{equation}
\bS'
= \begin{pmatrix}
 0.684 & -0.003 & -0.006 & -0.014 & 0.013 \\
 0.006 & 0.001 & -0.004 & 0.018 & -0.007 \\
 -0.008 & 0.003 & 0.001 & -0.007 & -0.018 \\
 -0.003 & -0.021 & 0.001 & -0.004 & -0.014 \\
 -0.005 & -0.002 & 0.017 & 0.001 & -0.001 
\end{pmatrix}
\end{equation}
from which we obtain a degree of order $\approx 0.684$ in agreement with the expected value of $0.69$.

\subsection{Biaxial phase}
\label{sec:examples_biaxial}
We now present an analogous analysis, but relative to a less symmetric phase, namely $D_{2h}$.
The ordering matrix we consider in this example is
\begin{equation}
\bS 
= \begin{pmatrix}
0.301 & -0.115 & -0.003 & 0.004 & -0.001 \\
0.127 & -0.537 & 0.007 & 0.000 & 0.002 \\
0.131 & -0.354 & 0.006 & 0.000 & -0.003 \\
0.403 & -0.303 & 0.003 & 0.001 & 0.004 \\
0.118 & 0.255 & 0.000 & 0.004 & -0.001 
\end{pmatrix}
\label{eq:bS_example_biaxial}
\end{equation}
which is built to represent a biaxial phase with symmetry axes rotated by $\alpha=60^\circ$, $\beta=30^\circ$,
$\gamma=45^\circ$ with respect to the laboratory axes. The order parameters, i.e. the ordering matrix entries 
when referred to its principal axes, are $S_{00} \approx 0.507$, $S_{01} \approx -0.179$, $S_{10} \approx -0.201$, 
$S_{11} \approx 0.743$. In more standard notation, these order parameters correspond respectively to $S$, $D$, $P$ and $C$, 
albeit with different normalisation coefficients.

We sample the orientations of the laboratory frame uniformly, and find the following discrepancy coefficients 
for the five symmetry classes
\begin{center}
\begin{tabular}{c|ccccc}
$p$ & $0$ & $1$ & $2$ & $3$ & $4$ \\[1mm] \hline
\vspace{1mm}$c(p)$ & 1 & 0.43 & 0.070 & 0.017 & 0
\end{tabular}
\end{center}

The projection on the Orthorhombic class, yields $\alpha \approx 60.7^\circ$, $\beta \approx 29.7^\circ$ 
and $\gamma \approx 43.6^\circ$. The reconstruction of the ordering matrix with respect to its principal axes reads
\begin{equation}
\bS'
= \begin{pmatrix}
0.505 & -0.181 & 0.000 & 0.005 & 0.001 \\
-0.204 & 0.742 & -0.009 & 0.001 & -0.002 \\
0.06 & 0.028 & 0.000 & 0.000 & 0.005 \\
-0.005 & -0.005 & -0.005 & -0.001 & 0.001 \\
0.000 & 0.000 & 0.001 & -0.003 & 0.002 
\end{pmatrix}.
\end{equation}

\section{Conclusions}
\label{sec:conclusions}

We have proposed a method able to provide the canonical form of an ordering tensors $\bS$.
This canonical representation 
readily yields the order parameters, i.e., the scalar quantities that are usually 
adopted to describe the order 
in a liquid crystal compound. However, the physical meaning of the matrix entries 
(for low symmetry molecules and phases) 
is still a matter of debate and fall outside the scope of our paper. The laboratory axes 
of the canonical form are interpreted as
``directors'' and provide the directions of the symmetry axes, if there are any. Finally, 
we have shown that there are only five
possible phase symmetry-classes, when the orientational probability density function is 
truncated at the second-rank level. This is a standard approximation in many theoretical studies of
uniaxial and biaxial liquid crystals. 

Our proposed method is simple enough to be applicable to the analysis of real situations. 
For the examples considered in Sec.~\ref{sec:examples} the proposed algorithm seems to
be reliable and give a fast analysis of the ordering tensor that leads to the correct identification 
of a uniaxial and a biaxial phase in a model system. For this purpose it has been sufficient to
implement a very simple Monte-Carlo optimisation procedure. However, it may be appropriate to 
develop more efficient methods in case of more complex real systems.

Our strategy, based on the second-rank ordering tensor, is not able to distinguish amongst the phase groups belonging to the same class. In principle our approach could be easily extended 
to include higher-rank ordering tensors. The same mathematical ideas and tools we have
put forward could be applied to this more general case, only at the cost of more involved notations.
However, we believe that in so
doing the presentation and the readability of the paper could be seriously affected. 
Furthermore, the second-rank case is the most
relevant from a physical perspective, since most tensorial properties that can be measured in liquid crystals
are second-rank. For these reasons we have only given the detailed presentation in the case of 
a second-rank ordering tensor.

\section*{Acknowledgements} 
\comm{The authors wish to thank the two anonymous referees for their valuable comments, 
which led to an improved paper.} S.T. wishes to thank Maurizio Vianello and Antonio DiCarlo 
for instructive conversations concerning related problems in Elasticity. 
The authors are also grateful to the Isaac Newton Institute, Cambridge, where this work was originated during 
the Programme on the Mathematics of Liquid Crystals in 2013.

\appendix

\comm{\section{Appendix: orientational probability densities in $\Otre$}\label{apdx:otre}
When dealing with the orientation of molecules in the physical space, it is natural 
choose the molecular and laboratory \comm{frames of reference} with the same handedness.
The orientation of a molecule is then assigned in terms of the rotation $\bR\in \essotre$
that brings the laboratory axes into coincidence with the molecular axes. In this respect, 
an inversion or a mirror reflection do not represent a change of the orientation of the molecule and
the orientational probability density function is usually taken to be a function $f:\essotre \to \R_+$.

However, when we need to consider how the symmetry groups acts on the orientational 
distribution, for instance because we need to exploit the mirror symmetry of a molecule, 
we need to consider probability densities $g:\Otre \to \R_+$. The two pictures can be reconciled as follows. 

Since $\Otre = \essotre \times C_i$ is composed of two connected components, 
that is $\essotre$ and the other obtained from $\essotre$ by inversion, the integration over $\Otre$ separates into two integrals over $\essotre$. Namely, the ensemble average 
of a function $\chi:\Otre \to \R$ is
\begin{equation}
\langle \chi \rangle_{\Otre} = \int_{\essotre} \chi(\bR)g(\bR) \ \dd\mu(\bR) +
\int_{\essotre} \chi(\bR\iota) g(\bR\iota) \ \dd\mu(\bR).
\label{eq:averageO3}
\end{equation}

For \emph{apolar molecules}, which posses inversion symmetry, $g(\bR)=g(\bR\iota)$. 
The compatibility of the distributions in $\essotre$ requires
\begin{equation}
g(\bR) = \tfrac{1}{2}f(\bR), \qquad \text{ for all } \bR \in \essotre,  
\end{equation}
where $f(\bR)$ is the distribution function in $\essotre$ that we have defined in the text and the factor $1/2$ 
comes from the normalisation of the distributions. The ensemble average \eqref{eq:averageO3} then becomes
\begin{equation}
\langle \chi \rangle_{\Otre} = \int_{\essotre} \tfrac{1}{2}\left(\chi(\bR) + \chi(\bR\iota) \right) g(\bR) \ \dd\mu(\bR).
\label{eq:averageO3_b}
\end{equation}
This equation shows that the order parameter tensors are again calculated as averages over $\essotre$. Explicitly, we have
\begin{equation}
\langle D^{(j)} \rangle_{\Otre} = \int_{\essotre} \tfrac{1}{2}\left(D^{(j)}(\bR) + (-1)^j D^{(j)}(\bR) \right) g(\bR) \ \dd\mu(\bR).
\label{eq:opO3}
\end{equation}
where we have used
\begin{equation}
D^{(j)}(\bR\iota) = (-1)^j D^{(j)}(\bR), 
\end{equation}
an identity that follows from the generalisation of Eq. \eqref{eq:conjugation}
to $j^{\text{th}}$-rank tensors (see Eq.\eqref{eq:diagonal_action}). 
In particular, this shows that the order parameters for apolar molecules vanish if $j$ is odd, 
the second rank order parameter $\bS$ considered in the text clearly do not vanish since $j$ is even.

By contrast, for \emph{polar molecules}, which do not posses inversion symmetry, it is not possible 
(if the system is homogeneous, of the same chirality) to find an inverted molecule. 
Therefore, $g(\bR\iota) = 0$, for all $\bR \in \essotre$ and the probability density $g(\bR)$
coincides with the $f(\bR)$ in $\essotre$.
 
}

%\nocite{*}
% \bibliographystyle{ieeetr}
% \bibliography{references}

%\begin{thebibliography}{99}%

%\end{thebibliography}

\end{document}